\documentclass[conference]{IEEEtran}
\IEEEoverridecommandlockouts
\usepackage{epsfig,psfrag}
\usepackage[cmex10]{amsmath}
\usepackage{amsthm,graphics,amscd,amssymb,enumerate,xspace,latexsym,amsbsy,MnSymbol,stmaryrd,wrapfig}

\usepackage{pgf}
\usepackage{tikz}
\usetikzlibrary{arrows,automata,backgrounds,decorations}
\usetikzlibrary{shapes.multipart}
\usepgflibrary{decorations.pathreplacing}
\usetikzlibrary{arrows,calc,topaths}
\usetikzlibrary{automata,shapes.arrows, backgrounds, fadings,intersections, positioning}
\usetikzlibrary{shadows}
\tikzstyle{small}=[font=\footnotesize]
\tikzset{
    every picture/.style={>=stealth,auto,node distance=2cm},
}

\usetikzlibrary{petri}
\usetikzlibrary{shadows}
\usetikzlibrary{decorations.pathmorphing}
\usetikzlibrary{decorations.shapes}
\usetikzlibrary{shapes}
\usetikzlibrary{backgrounds}
\usetikzlibrary{fit}
\usetikzlibrary{arrows}
\usetikzlibrary{calc}
\usetikzlibrary{mindmap}
\usetikzlibrary{matrix}

\usepackage[capitalise,noabbrev,nameinlink]{cleveref}

\newtheorem{theorem}{Theorem}
\newtheorem{lemma}[theorem]{Lemma}

\newtheorem{proposition}[theorem]{Proposition}
\newtheorem{definition}{Definition}

\newcommand{\hide}[1]{}

\usepackage{centernot}

\newcommand{\Act}{\mathit{Act}}

\newcommand{\ignore}[1]{}

\newcommand{\N}{\mathbb{N}}

\newcommand{\Z}{\mathbb{Z}}
\newcommand{\C}{\mathbb{C}}
\newcommand{\x}{\times}

\newcommand{\R}{Spoiler}
\newcommand{\Scolour}{black}
\newcommand{\V}{Duplicator}
\newcommand{\Dcolour}{white}

\newcommand{\step}[1]{\Step{#1}{}{}}

\newcommand{\Step}[3]{\ensuremath{\,{\stackrel{#1}{\longrightarrow}}\!{}^{\scriptstyle{#2}}_{\scriptstyle{#3}}}\,}

\newcommand{\SIMSYMBOL}{\preceq}

\newcommand{\SIM}[2]{\ensuremath{\mathrel{\SIMSYMBOL^{#1}_{#2}}}}
\newcommand{\notSIM}[2]{\ensuremath{\mathrel{\not\SIMSYMBOL^{#1}_{#2}}}}

\newcommand{\ssim}{\SIMSYMBOL}

\newcommand{\nic}[1]{}

\begin{document}

\title{Infinite-State Energy Games\thanks{This is the technical report number EDI-INF-RR-1419 of
the School of Informatics at the University of Edinburgh,
UK. (http://www.inf.ed.ac.uk/publications/report/). 
Full version (including proofs) of material presented at CSL-LICS 2014 (Vienna, Austria). arXiv.org - CC BY 3.0.}}

\author{
\IEEEauthorblockN{Parosh Aziz Abdulla\IEEEauthorrefmark{1},
Mohamed Faouzi Atig\IEEEauthorrefmark{1},
Piotr Hofman\IEEEauthorrefmark{2},
Richard Mayr\IEEEauthorrefmark{4},
K.~Narayan Kumar\IEEEauthorrefmark{3},
Patrick Totzke\IEEEauthorrefmark{4}}
\IEEEauthorblockA{\IEEEauthorrefmark{1}Uppsala University, Sweden}
\IEEEauthorblockA{\IEEEauthorrefmark{2}University of Bayreuth, Germany}
\IEEEauthorblockA{\IEEEauthorrefmark{3}Chennai Mathematical Institute, India}
\IEEEauthorblockA{\IEEEauthorrefmark{4}University of Edinburgh, UK}
}

\maketitle

\begin{abstract}
Energy games are a well-studied class of 
2-player turn-based games on a finite graph
where transitions are labeled with integer vectors which represent 
changes in a multidimensional resource (the energy). One player tries to
keep the cumulative changes non-negative in every component while
the other tries to frustrate this.

We consider generalized energy games played on
infinite game graphs induced by pushdown automata (modelling recursion)
or their subclass of one-counter automata.

Our main result is that energy games are decidable in the case where the game 
graph is induced by a one-counter automaton and the energy is one-dimensional.
On the other hand, every further generalization is undecidable:
Energy games on one-counter automata with a 2-dimensional energy are
undecidable, and energy games on pushdown automata are undecidable even if the
energy is one-dimensional.

Furthermore, we show that energy games and simulation games are
inter-reducible, and thus we additionally obtain several 
new (un)decidability results for the problem of checking simulation
preorder between pushdown automata and vector addition systems.
\end{abstract}
\noindent

\section{Introduction}\label{sec:introduction}

Two-player turn-based games on transition graphs provide the mathematical foundation for the 
analysis of reactive systems, and they are used to solve many problems in
formal verification, e.g., in model checking and semantic equivalence checking
\cite{Stirling:IGPL99}.
The vertices of the game graph represent states of the system, and they 
are partitioned into subsets that belong to Player $0$ and
Player $1$, respectively. The game starts at an initial vertex, and in every
round of the game the player who owns the current vertex chooses an outgoing
transition leading to the next vertex. This yields an either finite or
infinite sequence of visited vertices, called a {\em play} of the game.

Various types of games define different winning conditions that classify a
play as winning for a given player, e.g., reachability, safety, liveness,
$\omega$-regular, or parity objectives.

A generalization of such games introduces quantitative aspects and
corresponding quantitative winning conditions. Transitions are labeled
with numeric values, typically integers, that are interpreted as the cost or
reward of taking this transition, e.g., elapsed time, lost/gained material
resources or energy, etc. The value of (a prefix of) a play is then defined
as the sum of the values of the used transitions.
Further generalizations use multi-dimensional labels (i.e., vectors of
integers) instead of single integers.

The most commonly studied quantitative games are {\em energy games}
and {\em limit-average games} (also called mean-payoff games), 
which differ in the quantitative winning
condition.

In energy games, the objective of Player $1$ is to forever keep the value of
the prefix of the play
non-negative (resp.\ non-negative in every component, for multidimensional
values), while Player $0$ tries to frustrate this. 
Intuitively, this means that the given resource (e.g., the stored energy) must never run
out during the operation of the system.
Clearly such games are monotone in the resource value, in the sense that
higher values are always beneficial for Player $1$.

There are two classic problems about energy games.
In the {\em fixed initial credit problem} 
one asks whether Player $1$ has a winning strategy from a given starting
configuration with a fixed initial energy (resource value).
In the {\em unknown initial credit problem} one quantifies over this initial
energy and asks whether there exists a sufficiently high value for Player $1$
to win. Even if the answer is positive,
this does not yield any information about the minimal initial energy required.

In limit-average games, the objective is to maximize the average value per
step of the
play in the long run. I.e., one asks whether Player $1$ has a strategy to
keep the average value per transition above a given number $k$ in the long
run. Limit-average games are closely related to the unknown initial credit
problem in energy games, since in both cases one tries to maximize the payoff
in the long run without considering short-term fluctuations.
(The fixed initial credit problem in energy games is different however,
since local fluctuations matter.)

{\bf Previous work on finite game graphs.}
Most previous work on quantitative games has considered energy games and
limit-average games on finite game graphs, sometimes combined with classic
winning conditions such as parity objectives.
The  unknown and fixed initial credit problems for one-dimensional 
energy parity games were shown to be decidable in
\cite{Chatterjee-Doyen:TCS2012}. 
The unknown initial credit problem for multidimensional energy parity games 
is known to be coNP-complete \cite{Chatterjee:FSTTCS2010,CRR:CONCUR2012}. 
The fixed initial credit problem for $n$-dimensional energy games can be
solved in $n$-EXPTIME \cite{icalp:BrazdilJK10}, and the 
fixed initial credit problem for multidimensional energy parity games is 
decidable \cite{AMSS:CONCUR2013}. An EXPSPACE lower bound follows by a
reduction from Petri net coverability \cite{lipton}.
Multidimensional limit-average games are coNP-complete
\cite{Chatterjee:FSTTCS2010}.

{\bf Previous work on infinite game graphs.}
Pushdown automata have been studied extensively as
a model for the analysis of recursive programs (e.g.,
\cite{BEM97,EHRS00,Walukiewicz00,Walukiewicz01}).
Due to the unbounded stack memory, they typically induce infinite transition graphs.
Two-player reachability, B\"uchi and parity games on pushdown automata are
EXPTIME-complete \cite{PitermanV04,Cachat02,Walukiewicz00,Walukiewicz01}.
A strict subclass of pushdown automata are one-counter automata, i.e., Minsky
machines with a single counter. They correspond to pushdown automata with only
one stack symbol (plus a non-removable bottom stack symbol).
Two-player reachability, B\"uchi and parity games on one-counter automata are
PSPACE-complete \cite{Serre:FOSSACS06,MGT:LICS2009}.

Quantitative extensions of pushdown games with limit-average objectives have
been studied in
\cite{Chatterjee-Velner:LICS2012,Chatterjee-Velner:CONCUR2013}. 
The authors show the undecidability of limit-average pushdown games with one
resource-dimension, by reduction from the non-universality problem 
of weighted finite automata \cite{DBLP:conf/atva/AlmagorBK11}. On the other
hand, they prove the decidability of limit-average pushdown games under modular
strategies (a restriction on how the resources interact with the recursion).

{\bf Our contribution.}
We consider energy games that are played on infinite game graphs that
are induced by either pushdown automata or one-counter automata.
We consider both single-dimensional and multi-dimensional energies and focus
on the fixed initial credit problem.

Our first observation is that energy games are closely connected to simulation
games, i.e., to checking simulation preorder between transition graphs.
Energy games on pushdown automata (resp.\ one-counter automata) with
$n$-dimensional energy are inter-reducible with checking simulation preorder
between pushdown automata (resp.\ one-counter automata) and $n$-dimensional
vector addition systems with states (VASS; aka Petri nets).
Using this connection, we show several (un)decidability results for
infinite-state energy games. 

We show that the winning sets in
single-dimensional energy games on one-counter automata are semilinear
by establishing semilinearity of the corresponding simulation game.
This yields a positive semi-decision procedure for the fixed initial credit
problem by the decidability of Presburger arithmetic. Since negative
semi-decidability is easily achieved by unfolding the game tree,
we obtain the decidability of the fixed initial credit problem 
of single-dimensional energy games on one-counter automata.

Moreover, we show that every further generalized infinite-state energy game is
undecidable by reduction from the halting problem for Minsky machines.
The fixed initial credit problem for 2-dimensional energy games on one-counter 
automata is undecidable.
For energy games on pushdown automata, both the fixed and the unknown initial
credit problem are undecidable, even if the energy is just single-dimensional.

\newpage
\section{Preliminaries}\label{sec:preliminaries}

Let $\Z$ denote the integers and $\N$ the non-negative
integers.

\begin{definition}\label{def:lts}
A \emph{labeled transition system} is described by a triple
$T=(V,\Act,\step{})$ where $V$ is a (possibly infinite) set of
states, $\Act$ is a finite set of action labels and $\step{}\subseteq V\x \Act\x V$ 
is the labeled transition relation. We use the infix notation $s\step{a}s'$ for a
transition $(s,a,s')\in\step{}$, in which case we say $T$ makes an \emph{$a$-step} from $s$ to $s'$. 
In the context of games, $V = V_0 \cup V_1$ is partitioned into the subset $V_0$ of states that belong 
to player $0$ and $V_1$ of states that belong to player $1$.
\end{definition}

\begin{definition}\label{def:pda}
A \emph{pushdown automaton} $A=(Q,\Gamma,\Act,\delta)$
is given by a finite set of control-states $Q$, a finite stack alphabet $\Gamma$,
a finite set of action labels $\Act$, and a finite set of transitions 
$\delta \subseteq Q \x \Gamma \x \Act \x Q \x \Gamma^*$.
It induces a transition system over $V = Q \x \Gamma^+$ by
$qX\alpha \step{a} q'\beta\alpha$ iff $(q,X,a,q',\beta) \in \delta$, for any
$\alpha \in \Gamma^*$.
\end{definition}

\begin{definition}\label{def:peg}
A \emph{pushdown energy game} of dimension $n$ between two players $0$ and $1$ 
is given by $G=(Q_0, Q_1,\Gamma,\delta,n)$.
$Q_0$ and $Q_1$ are finite sets of control-states that belong to player $0$ 
and $1$, respectively. $\Gamma$ is a finite stack alphabet
and $n \in \N$ represents the dimension of the energy.

The transition relation 
$\delta \subseteq (Q_0\cup Q_1) \x \Gamma \x (Q_0 \cup Q_1) \x \Gamma^*
\x \{-1,0,1\}^n$ induces a game graph over $(Q_0 \cup Q_1) \x \Gamma^+ \x \Z^n$ as follows:
If $(q,X,q',\beta,C) \in \delta$ then
$(q,X\alpha,E) \step{} (q',\beta\alpha, E')$ 
for any $\alpha \in \Gamma^*$ and $E' = E+C$.
(We don't use transition labels, since they do not influence the
semantics of energy games.) 

In the induced game graph, configurations with control-states in $Q_0$ and $Q_1$ 
belong to player $0$ and $1$, respectively. The player who owns the current configuration
gets to choose the next step. Without restriction, we assume that every
configuration has at least one outgoing transition.
The game stops and Player $0$ wins if a configuration 
$(q,\alpha,E)$ is reached where 
$E=(e_1,\dots,e_n)$ with $e_i <0$ for some $i$.
Player $1$ wins every infinite game.
\end{definition}

In the special case where the stack is never used, the pushdown energy game
corresponds to just an ordinary $n$-dimensional energy game with a finite
control-graph
\cite{CRR:CONCUR2012,AMSS:CONCUR2013}.
In particular, the energy dimensions are not genuine counters. They cannot be
tested for zero and never influence the available transitions, but only affect
the winning condition of the game. 

One-counter automata can be seen as a special subclass of pushdown automata where
there is only one stack symbol plus a non-removable stack bottom symbol $\bot$. 
In order to keep the presentation clear, we define an explicit notation for
one-counter automata and games.

\begin{definition}\label{def:oca-ocn}
A \emph{one-counter automaton (OCA)} $A=(Q,\Act,\delta,\delta_0)$
is given by a finite set of control-states $Q$, a finite set of action labels $\Act$
and transition relations $\delta\subseteq Q\x \Act\x\{-1,0,1\}\x Q$
and $\delta_0\subseteq Q\x \Act\x\{0,1\}\x Q$.
Such an automaton is called a \emph{one-counter net (OCN)} if
$\delta_0=\emptyset$, i.e., if the automaton 
cannot test if the counter is equal to $0$.

These automata induce an infinite-state labeled transition system over the
stateset $Q\x\N$, whose elements will be written as $pm$, 
and transitions are defined as follows:
$pm\step{a}p'm'$ iff
        \begin{enumerate}
          \item $(p,a,d,p')\in\delta$ and $m'=m+d \geq 0$ or
          \item $(p,a,d,p')\in\delta_0$, $m=0$ and $m'=d$.
        \end{enumerate}
\end{definition}

\begin{definition}\label{def:ocaeg}
A \emph{one-counter energy game} of dimension $n$ between two players $0$ and $1$ 
is given by $G=(Q_0, Q_1,\delta,\delta_0,n)$.
$Q_0$ and $Q_1$ are finite sets of control-states that belong to player $0$ 
and $1$, respectively, and $n \in \N$ represents the dimension of the energy.
The transition relations $\delta\subseteq (Q_0 \cup Q_1)\x\{-1,0,1\}\x (Q_0 \cup Q_1) 
\x \{-1,0,1\}^n$
and $\delta_0\subseteq (Q_0 \cup Q_1)\x\{0,1\}\x (Q_0\cup Q_1) \x \{-1,0,1\}^n$
induce an infinite game graph over $(Q_0\cup Q_1)\x\N \x \Z^n$ as follows.
The number $m \in \N$ represents the value of the one genuine testable counter,
while the $E \in \Z^n$ represents the available multidimensional energy. 
We have $(p,m,E) \step{} (p',m',E')$ iff
        \begin{enumerate}
          \item $(p,d,p',C)\in\delta$ and $m'=m+d \geq 0$ and $E' = E+C$ or
          \item $(p,d,p',C)\in\delta_0$, $m=0$ and $m'=d$ and $E' = E+C$.
        \end{enumerate}
The players choose moves depending on who owns the current control-state.
The game stops and Player $0$ wins if a configuration 
$(q,k,E)$ is reached where 
$E=(e_1,\dots,e_n)$ with $e_i <0$ for some $i$.
Player $1$ wins every infinite game.
\end{definition}

\begin{definition}\label{def:VASS}
A \emph{vector addition system with states (VASS)} of dimension $n$ is
given by $(Q,\Act,\delta)$. $Q$ is a finite set of control-states,
$\Act$ is a finite set of action labels and 
$\delta \subseteq Q \x \Act \x Q \x \{-1,0,1\}^n$ is a finite transition relation.
It induces an infinite-state labeled transition system over the
stateset $Q\x\N^n$ as follows.
We have $(p,C) \step{a} (p',C+D)$ iff 
there is some $(p,a,p',D) \in \delta$ s.t.\ $C+D \in \N^n$.

A VASS of dimension $1$ corresponds to an OCN.
\end{definition}

{\bf\noindent Problems about energy games.}
Previous works on energy games (on a finite game graph)
mainly considered the following two problems
\cite{Chatterjee:FSTTCS2010,Chatterjee-Doyen:TCS2012,CRR:CONCUR2012,AMSS:CONCUR2013}.

In the {\em fixed initial credit problem}, one considers 
a starting configuration with a fixed initial energy.
The question is whether Player $1$ has a winning strategy in the energy game,
starting from this configuration with the given amount of energy.

In the {\em unknown initial credit problem} one instead asks whether
there {\em exists} some level of initial energy s.t.\ Player $1$ can win the game.
Even if the answer to this question is
positive, it does not necessarily yield any information about the minimal initial
energy required to win the game.

%
{\bf\noindent Outline of the results.}
In the following section we show
that there is a general connection between energy games and simulation 
games (i.e., checking simulation preorder between transition graphs). 
Pushdown energy games of energy dimension $n$ are logspace inter-reducible
with simulation games between a pushdown automaton and an $n$-dimensional
vector addition system with states (VASS). A similar result holds for one-counter
energy games and simulation between OCA and $n$-dimensional VASS.

Using this connection, we prove several decidability results for energy games
on infinite game graphs.
\begin{enumerate}
\item
For one-counter energy games of energy dimension $n=1$
the winning sets are semilinear and the fixed initial credit problem
is decidable.
\item
The fixed initial credit problem is undecidable for
one-counter energy games of energy dimension $n\ge 2$.
\item
Both the fixed and the unknown initial credit problem are undecidable
for pushdown energy games, even for energy dimension $n=1$.
\end{enumerate}

\section{Energy Games vs.\ Simulation Games}\label{sec:simulation}

Simulation is a semantic preorder in van Glabbeeks
linear time -- branching time spectrum \cite{Gla2001}.
It is used to compare the behavior of processes
and is defined as follows.

\begin{definition}[Simulation]\label{def:sim}
Given two labeled transition systems $T$ and $T'$,
a relation $R$ on the disjoint union of the sets of states of $T$
and $T'$ is a \emph{simulation} if for every 
pair of states $(c, c') \in R$ and every step $c \step{a} d$
there exists a step $c' \step{a} d'$ such that $(d,d') \in R$.

Simulations are closed under union. So there exists a unique maximal simulation
$\SIM{}{T,T'}$ which is a preorder, commonly called \emph{simulation preorder}.
We drop the index whenever it is clear from the context and say
that $c'$ \emph{simulates} $c$ iff $c\ssim c'$.
By simulation between $M$ and $M'$ or w.r.t.~$M,M'$
we mean the maximal simulation $\SIM{}{T,T'}$ relative
to the transition systems $T$ and $T'$ which are induced by
$M$ and $M'$, respectively.

Simulation preorder can also be characterized in terms of an interactive game between two
players \R\ (Player $0$) and \V\ (Player $1$), 
where the latter tries to stepwise match the moves of the former.  
A \emph{play} is a finite or infinite sequence of pairs of 
states $(c_0,c'_0),(c_1,c'_1),\dots,(c_i,c'_i)\dots$ where the next pair
$(c_{i+1},c'_{i+1})$ is determined by a round of choices: First \R\ chooses a
transition $c_i\step{a}c_{i+1}$, then \V\ responds by choosing an equally
labeled transition $c'_i\step{a}c'_{i+1}$.
If one of the players cannot move then the other wins,
and \V\ wins every infinite play.
A \emph{strategy} is a set of rules that tells a player which valid move to choose.
A player plays according to a strategy if all his moves obey the rules of the strategy.
A strategy is \emph{winning} from $(c,c')$ if every play that starts in $(c,c')$
and which is played according to that strategy is winning.
We have $c\ssim c'$ iff \V\ has a winning strategy from $(c,c')$.
\end{definition}

First we show how energy games can be reduced to simulation games.

\begin{lemma}\label{lem:dec:reduction}
For any $n$-dimensional pushdown energy game $G = (Q_0, Q_1,\Gamma,\delta,n)$
one can in logspace construct a pushdown automaton $A=(Q_0 \cup Q_1 \cup Q_A,\Gamma,\Act,\delta_A)$
and a $n$-dimensional VASS $V = (Q_0 \cup Q_1 \cup Q_V,\Act,\delta_V)$ 
s.t., for every $q \in Q_0 \cup Q_1$, $\gamma \in \Gamma^+$ and $E \in \N^n$, 
Player $1$ wins the energy game from configuration $(q,\gamma,E)$ iff
$(q,\gamma) \ssim (q,E)$.

Moreover, in the special case of a one-counter energy game, the constructed automaton $A$ is a OCA.
\end{lemma}
\begin{proof}
Every step in the energy game on $G$ is emulated by either one or
two rounds of the simulation game between $A$ and $V$.
We maintain the invariant that a configuration $(q,\gamma,E)$ in the energy game
corresponds to a configuration $((q,\gamma),(q,E))$ in the simulation game.

For every transition $t = (q_t,X_t,q'_t,\beta_t,C_t) \in \delta$ and 
every stack symbol $X \in \Gamma$ we define unique action labels
$a_t, a_X \in \Act$. Moreover, we add symbol $a$ to $\Act$.
So $\Act = \{a_t, a_X\,|\, t \in \delta, X \in \Gamma\} \cup \{a\}$.

There are two cases, depending on which player chooses the step 
in the energy game from the current configuration $(q,\gamma,E)$, 
i.e., whether $q \in Q_0$ or $q \in Q_1$.

For $q \in Q_0$ and every transition $t = (q_t,X_t,q'_t,\beta_t,C_t) \in \delta$ with $q_t =q$
we add a transition $(q_t,X_t,a_t,q'_t,\beta_t)$ to $\delta_A$
and a transition $(q_t,a_t,q'_t,C_t)$ to $\delta_V$. Since the label $a_t$ is unique to the
transition $t$, Player $1$ has no choice in the simulation game but to implement
the effect of the transition $t$ chosen by Player $0$, and the invariant is
preserved.

For every $q \in Q_1$ we add a copied auxiliary control-state $\hat{q}$ to $Q_A$,
and we add transitions $(q,X,a_X,\hat{q},X)$ to $\delta_A$ for every $X \in \Gamma$.
I.e., in the simulation game Player $0$ makes a dummy-move that announces the current top
stack symbol $X$ via the action symbol $a_X$. The choice of the next
encoded transition $t$ (among those that are currently enabled in the energy game)
is made by Player $1$ in his next move.

For every transition $t = (q_t,X_t,q'_t,\beta_t,C_t) \in \delta$ with $q_t=q \in Q_1$ 
and $X_t = X \in \Gamma$ 
we add a transition
$(q_t,a_X,q''_t,C_t)$ to $\delta_V$ where $q''_t$ is a new auxiliary control-state that is
added to $Q_V$. I.e., Player $1$ gets to choose a transition $t$ from the current encoded
control-state of the energy game. Since he needs to use the same symbol $a_X$
as in the previous move by Player $0$, his choices are limited to transitions that are currently
enabled at control-state $q$ and top stack symbol $X$ in the energy game.
This choice of transition $t$ by Player $1$ is recorded
in the new control-state $q''_t$. Transitions that would decrease the energy below zero
(and thus be losing for Player $1$ in the energy game) 
are disabled by the semantics of VASS. 

In the next step, Player $0$ will be forced to implement this chosen transition 
$t$, or else she loses the game.
For every 
transition $t = (q_t,X_t,q'_t,\beta_t,C_t) \in \delta$ with $q_t=q \in Q_1$
we add a transition $(\hat{q},X_t,a_t,q'_t,\beta_t)$ to $\delta_A$.
This emulates a transition $t$ of the energy game and announces its unique
identity via the action symbol $a_t$. It remains to check whether this 
transition was the same as the one chosen by Player $1$ in the previous round.
(If not, then Player $0$ must lose the game.)

For every 
transition $t = (q_t,X_t,q'_t,\beta_t,C_t) \in \delta$ with $q_t=q \in Q_1$
we add a transition $(q''_t,a_t,q'_t,\vec{0})$ to $\delta_V$,
where $\vec{0}=\{0\}^n$ denotes the $n$-tuple with value $0$ on all
coordinates.
This simply implements the effect of the chosen transition $t$ in the case where Player $0$
has taken the correct transition with label $a_t$.
In the other case where Player $0$ did not choose the correct transition with label $a_t$
from state $\hat{q}$, we must ensure that Player $1$ wins the simulation game.
Thus we add transitions $(q''_t,b,u,\vec{0})$ to $\delta_V$
for every $b \neq a_t$ where $u$ is a universal winning state for Player $1$, i.e., 
we add a state $u$ to $Q_V$ and transitions $(u,c,u,\vec{0})$ to $\delta_V$
for every $c \in \Act$.

Thus, to avoid losing the simulation game, the players emulate the effect of 
an enabled energy game transition $t$ that was chosen by Player $1$ in two rounds of
the simulation game, and the invariant is preserved.

If Player $1$ has a strategy to win the energy game, then by faithful emulation
he also wins the simulation game. Since the encoded energy never drops below zero,
the corresponding transitions in the VASS are never blocked by the boundary condition
and the play of the simulation game is infinite.
Otherwise, if Player $0$ has a winning strategy in the energy game, then she 
can enforce that some dimension of the energy becomes negative.
By faithful emulation Player $0$ also wins the simulation game, since the steps 
that go below zero are blocked in the VASS where Player $1$ plays.  

Thus, for every $q \in Q_0 \cup Q_1$, 
Player $1$ wins the energy game from configuration $(q,\gamma,E)$ iff
he wins the simulation game from $((q,\gamma),(q,E))$ iff
$(q,\gamma) \ssim (q,E)$.

We observe that the above construction preserves the property that makes the pushdown automaton correspond
to an OCA. If there is only one stack symbol plus a non-removable stack bottom symbol in the 
pushdown energy game $G$ then the same property also hold for the constructed pushdown automaton $A$. 
Thus, one-counter energy games reduce to simulation games between OCA and VASS.
\end{proof}

For the reverse direction we show how to reduce simulation games to energy games.

\begin{lemma}\label{lem:dec:reverse-reduction}
For a pushdown automaton $A=(Q_A,\Gamma,\Act,\delta_A)$
and a $n$-dimensional VASS $V = (Q_V,\Act,\delta_V)$,
one can in logspace construct a 
$n$-dimensional pushdown energy game $G = (Q_0, Q_1, \Gamma,\delta,n)$
with $Q_0 = Q_A \x Q_V \x \{0\}$ and $Q_1 = Q_A \x Q_V \x \Act$
s.t., for every $q_0 \in Q_A$, $q_1 \in Q_V$, $\gamma \in \Gamma^+$ and $E \in \N^n$, 
we have $(q_0,\gamma) \ssim (q_1,E)$ iff
Player $1$ wins the energy game from configuration $((q_0,q_1,0),\gamma,E)$.

Moreover, if $A$ is a OCA then the constructed game $G$ is a one-counter energy game.
\end{lemma}
\begin{proof}
Every round of the simulation game is emulated by two steps in the energy game,
one step by Player $0$ followed by one step by Player $1$, and the stated invariant
will be preserved.

For every transition $(q_0,X,a,q_0',\beta) \in \delta_A$
and every state $q_1 \in Q_V$ we add a transition
$((q_0,q_1,0),X,(q_0',q_1,a),\beta,\vec{0})$ to $\delta$.
Since $(q_0,q_1,0) \in Q_0$, Player $0$ gets to chose this move in the energy game.
This move in the energy game does not change the energy, but 
it is enabled iff the corresponding move is enabled in the pushdown automaton and
it has the same effect on the stack. 
The new control-state $(q_0',q_1,a)$ records the new state $q_0'$ and the symbol $a \in \Act$,
which forces Player $1$ to emulate an $a$-move of the VASS in the next step.
Since $(q_0',q_1,a) \in Q_V$, Player $1$ chooses the next step.
For every transition $(q_1,a,q_1',D) \in \delta_V$, $q_0' \in Q_A$, $a \in \Act$
and $X \in \Gamma$, we add a transition 
$((q_0',q_1,a),X,(q_0',q_1',0),X,D)$ to $\delta$.
This move is enabled regardless of the stack content, but it must match the recorded 
control-state and action symbol. Its effect $D$ on the energy implements the changes
in the VASS. Unlike in the VASS, moves that decrease a counter below zero are not blocked
in the energy game, but they are losing for Player $1$. Thus, since Player $1$ chooses the
moves that affect the energy, he will avoid any move that is disabled in the VASS.
After this move we have emulated one round of the simulation game, the control-state is 
in $Q_0$ again and the invariant is maintained.

Finally, to ensure deadlock-freedom in the energy game, we add
transitions $(q,X,q,X,\vec{0})$ for every $q \in Q_0$ and $X \in \Gamma$
and transitions $(q,X,q,X,(-1,\dots,-1))$ for every $q \in Q_1$ and $X \in \Gamma$.
I.e., if Player $0$ was deadlocked in the pushdown automaton 
then she will loop forever in the energy game
without decreasing the energy, and thus Player $1$ wins.
If Player $1$ is deadlocked in the VASS, then the only available moves in the energy game
repeatedly decrease the energy until Player $1$ loses the game.

If $(q_0,\gamma) \ssim (q_1,E)$ then Player $1$ has a winning strategy in the 
simulation game. By the semantics of VASS he can continue forever without going below
zero in any of the VASS counters. By using the same strategy in the emulating energy game
he can continue forever without running out of energy and thus wins the energy game
from configuration $((q_0,q_1,0),\gamma,E)$.
Conversely, if $(q_0,\gamma) \not\ssim (q_1,E)$ then Player $0$ has a winning strategy in the 
simulation game which eventually leads to a configuration where Player $1$ is blocked.
By using the same strategy in the corresponding energy game,
eventually a configuration is reached where only energy decreasing moves remain available to Player $1$
and he loses the energy game.

We observe that the above construction preserves the property that makes the pushdown automaton correspond
to an OCA. If there is only one stack symbol plus a non-removable stack bottom symbol in the 
pushdown automaton $A$ 
then the same property also hold for the constructed pushdown energy game $G$.
Thus, simulation games between OCA and VASS reduce to one-counter energy games.
\end{proof}

\section{The main decidability result}\label{sec:decidability}

Here we consider 1-dimensional one-counter energy games. We show that
the fixed initial credit problem is decidable and that the winning sets are
semilinear. The proof first shows a corresponding result for
simulation preorder between a OCA and a OCN, and then applies 
the connection between simulation games
and energy games from Section~\ref{sec:simulation}.

The cornerstone in our argument is the following property of simulation
preorder (shown below).

\begin{theorem}\label{thm:dec:sim-is-slin}
    Simulation preorder $\SIM{}{A,A'}$
    between a given one-counter automaton
    $A$ and a one-counter net $A'$ is semilinear.
\end{theorem}

It immediately yields the decidability of simulation preorder.

\begin{theorem}\label{thm:sim-OCA-OCB-dec}
Simulation preorder between a OCA and a OCN is decidable.
\end{theorem}
\begin{proof}
We use a combination of two semi-decision procedures.
Since OCA/OCN define finitely-branching processes,
we can apply a standard result that \emph{non}-simulation
is semidecidable.
The semi-decision procedure for simulation works as follows.
By \cref{thm:dec:sim-is-slin}, it suffices to enumerate semilinear sets
and to check for each such set whether it is a simulation relation that moreover contains the given
pair of processes. This check is effective by the definition of the simulation
condition and the decidability of Presburger arithmetic.
\end{proof}

%
%
%

By the connection between simulation games
and energy games from Section~\ref{sec:simulation}, we obtain our main result.

\begin{theorem}\label{thm:dec:1D-oca-eg}
The fixed initial credit problem for 1-dimensional one-counter energy games
is decidable, and the winning sets are semilinear.
\end{theorem}
\begin{proof}
Directly by \cref{thm:dec:sim-is-slin}, \cref{thm:sim-OCA-OCB-dec}, 
\cref{lem:dec:reduction}.
\end{proof}

In the rest of this section we prove \cref{thm:dec:sim-is-slin}.
We fix a one-counter automaton
$A=(Q,\Act,\delta,\delta_0)$ and a one-counter net $A'=(Q',\Act,\delta')$.
Note that a slightly more general problem, where \emph{both} systems have zero-testing,
is no longer computable. Simulation preorder between two one-counter \emph{automata}
is undecidable \cite{JMS1999}.

In our construction, we will use a previous result about a 
special subcase of our problem, that the maximal simulation between two
one-counter \emph{nets} is effectively semilinear \cite{JKM2000,HLMT:FSTTCS2013}.
Ultimately, these positive results are due to the following monotonicity
properties.

\begin{proposition}[Monotonicity]\label{lem:dec:monotonicity}
    Let $p$ be a state of a OCA,
    $p',q'$ be states of a OCN and $m,m',n',l\in \N$. Then,
    \begin{enumerate}
        \item $p'm'\step{a}q'n'$ implies $p'(m'+l)\step{a}q'(n'+l)$,
        \item $p'm'\ssim p'(m'+l)$ and
        \item if $pm\ssim p'm'$ then $pm\ssim p'(m'+l)$.
    \end{enumerate}
\end{proposition}

Following \cite{JMS1999}, we interpret a binary relation
$R\subseteq (Q\x\N)\x(Q'\x\N)$
between the configurations of the processes of $A$ and $A'$ as 
a $2$-coloring of $|Q\x Q'|$ many planes $\N\x\N$,
one for every pair $(p,p')$ of control states.
The color of $(m,m')$ on the plane for $(p,p')$
is \Dcolour\ if $(pm,p'm')\in R$ and \Scolour\ otherwise.
We are particularly interested in the coloring of $\ssim$,
the largest simulation w.r.t.~$A$ and $A'$.

\newcommand{\colour}[3]{\C_{#1}(#2,#3)}
\newcommand{\Cline}[2]{\colour{#1}{#2}{-}}
\newcommand{\minsuff}[2]{W_{#1}(#2)}

\begin{definition}\label{def:dec:colouring}
    Consider the \emph{coloring} $\C$
    defined by
    $\colour{p,p'}{m}{m'} = \Dcolour$ iff $pm\ssim p'm'$.
    We write $\Cline{p,p'}{i}$ for the vertical line at level $i$
    on the plane for states $(p,p')$. That is,
    $\Cline{p,p'}{i}:\N\to\{\Dcolour,\Scolour\}$ with
    $\Cline{p,p'}{i}(n)=\colour{p,p'}{i}{n}$.
    We say that this line is \Scolour\ iff $\Cline{p,p'}{i}(n)=\Scolour$ for all $n\in\N$,
    i.e., if every point on the line is colored \Scolour.
    
    By monotonicity of simulation preorder (Proposition \ref{lem:dec:monotonicity}.3), for every line that is \emph{not}
    \Scolour\ there is a minimal value $\minsuff{p,p'}{i}$
    such that $\Cline{p,p'}{i}(n)=\Dcolour$ for all $n\ge\minsuff{p,p'}{i}$.
    We define $\minsuff{p,p'}{i}=\infty$ if no such value exists (the line is
    \Scolour)
    and write $\minsuff{}{i}$ for the maximal finite value $\minsuff{p,p'}{i}$
    over all pairs $(p,p')\in (Q\x Q')$.
\end{definition}

We show that the distribution of \Scolour\ lines in the coloring
of $\ssim$ follows a regular pattern.

\begin{definition}[Safe strategies]\label{def:dec:safe}
  Let $\sigma$ be a winning strategy for \R\ in the simulation game from position
  $(pm,p'm')$ for some $m,m',l \in \N$ such that $m\ge l$. 
  A strategy $\sigma$ is called \emph{$l$-safe},
  if whenever a play according to $\sigma$ reaches
  a position of the form $(ql,q'n')$ for the first time,
  then the line $\Cline{q,q'}{l}$ is \Scolour.
\end{definition}
%

\begin{lemma}\label{lem:dec:safe}
    Let $m\ge l\in\N$ such that $\Cline{p,p'}{m}$ is \Scolour.
    For all values $m'\in\N$,
    \R\ can win the simulation game from position $(pm,p'm')$ using a $l$-safe strategy.
\end{lemma}
\begin{proof}
  Fix any position $(pm,p'm')$. As the line $\Cline{p,p'}{m}$ is \Scolour, we have
  $pm\notSIM{}{}pn'$ for all $n'\in\N$. In particular, \R\ has a winning strategy
  $\sigma_l$ from position $(pm,p'(m'+\minsuff{}{l}))$.
  We see that \R\ may re-use this strategy also in the game that starts from position
  $(pm,p'm')$, maintaining an offset of $\minsuff{}{l}$ in her opponents counter value
  to the corresponding position in $\sigma_l$.
  Let $\sigma$ be the so constructed strategy
  for the game from $(pm,p'm')$.
  
  We argue that $\sigma$ is winning and $l$-safe.
  Indeed, let $(ql,q'n')$ be some position on a branch of $\sigma$.
  The same branch in $\sigma_l$ leads to position
  $(ql,q'(n'+\minsuff{}{l}))$. Clearly, $n'+\minsuff{}{l}\ge \minsuff{q,q'}{l}$ and the color
  of this point is \Scolour, because it is a position
  on a winning strategy for \R. Recall that \Scolour\ point on the line $\Cline{p,p'}{l}$ above $\minsuff{}{l}$ means that $\Cline{p,p'}{l}$ is \Scolour. 
  Together this means that the whole line $\Cline{q,q'}{l}$ must be \Scolour.
\end{proof}
%

\begin{lemma}\label{lem:dec:regularity}
    There exist $l,K\in\N$ such that
    for any pair $(p,p')\in (Q\x Q')$ of control-states
    and any $i\ge l$ it holds that
    \begin{enumerate}
        \item The line $\Cline{p,p'}{i}$ is \Scolour\ iff $\Cline{p,p'}{i+K}$ is
              \Scolour
        \item $\minsuff{p,p'}{i} \le \minsuff{p,p'}{i+K}$.
    \end{enumerate}
\end{lemma}
\begin{proof}
 We consider the patterns
 $Pat_i: Q\x Q'\to \{\Scolour,\Dcolour\}$ indicating the colors of
 all lines at level $i\in N$:
 $Pat_i(p,p') = \Scolour$ iff $\Cline{p,p'}{i}$ is \Scolour.
 Naturally, with increasing index $i$, this pattern eventually repeats
 and we can extract an infinite sequence of indices with the same pattern.
 This prescribes a sequence of vectors 
 in which each component corresponds to $\minsuff{q,q'}{i}<\infty$
 for some $(q,q')\in Q\x Q'$.
 Dickson's Lemma allows us to pick indices $l$ and $l +K\in\N$ that
 satisfy both conditions in the claim of the lemma
 (for $i=l$).
 It remains to show that
 both claims hold also for all $i>l$.

 For the first claim, assume towards a contradiction that for some pair
 $(p,p')$, the color of $\Cline{p,p'}{i}$ is different from that of
 $\Cline{p,p'}{i+K}$.
 W.l.o.g.\ assume further $\Cline{p,p'}{i}$ is \Dcolour\ (the other case is symmetric).
 Then,

 \begin{enumerate}
   \item $\forall m'.\ pi+K\notSIM{}{}p'm'$, and 
   \item $\forall m'\ge \minsuff{}{i}.\ pi\ssim p'm'$. 
 \end{enumerate}

 We fix $m'\ge \minsuff{}{i}$, that is, $pi\ssim p'm'$ and $pi+K\notSIM{}{}p'm'$.
 By \cref{lem:dec:safe}, \R\ has a $(l+K)$-safe winning strategy $\sigma$ for
 the simulation game from position $(p(i+K),p'm')$.
 \Cref{fig:dec:regularity} illustrates this scenario.

 We claim that \R\ can reuse this strategy and win also from position $(pi,p'm')$.
 Consider the simulation game from $(pi,p'm')$ in which \R\ initially plays according to $\sigma$.
 Then on any branch, either no position visits level $l$ or there is some first position
 which does. Fix some branch on this partial strategy.
 In the first case, the corresponding branch on $\sigma$
 never visits level $l+K$ and ends in some position $(q(n+K),q'n')$
 which is immediately winning for \R. This means our branch
 ends in position $(qn,q'n')$ where $n>l\ge 0$ and from which \R\ wins immediately
 because she can mimic the attack in $\sigma$.
 Alternatively, this branch visits some position $(ql,q'n')$
 at level $l$ for the first time
 and the corresponding branch in $\sigma$ visits $(q(l+K),q'n')$.
 Since $\sigma$ is $(l+K)$-safe, we know that the line $\Cline{q,q'}{l+K}$
 is \Scolour. Because of our assumption that the pattern at levels $l$ and $l+K$
 agree, the line $\Cline{q,q'}{l}$ must also be \Scolour.
 Therefore in particular we have $ql\notSIM{}{}q'n'$, so \R\ may continue
 the game from position $(ql,q'n')$ using some winning strategy.
 We have shown that there is a winning strategy for \R\ from position $(pi,p'm')$.
 As $m'$ was chosen arbitrarily, the line $\Cline{p,p'}{i}$
 is \Scolour, which contradicts our assumption and thus completes the proof
 of the first claim.

 For the second claim,
 it suffices to show that for all pairs $(p,p')\in Q\x Q'$, $i \ge l$ and $m'\in\N$ it holds that
 \begin{equation}
   p(i+K) \ssim p'm' \implies pi \ssim p'm'
 \end{equation}

 Let $\sigma$ be a winning strategy for \V\ in the simulation game
 from position $(p(i+K),p'm')$
 and consider a first position $(q(l+K),q'n')$ on a play of $\sigma$
 where \R's counter drops to $l+K$. Since this is a position on a winning
 strategy, $n'$ must be greater or equal $\minsuff{q,q'}{l+K}$.
 If \V\ plays according to the same strategy from position $(pi,p'm')$,
 then this partial play ends in $(ql,q'n')$,
 which is winning for \V\ because $n'\ge \minsuff{q,q'}{l+K}\ge
 \minsuff{q,q'}{l}$ due to our choice of $l$ and $K$.
 We conclude that $\sigma$ must be a winning strategy for \V\
 in the simulation game from $(pi,p'm')$.
\end{proof}
\begin{figure}
\begin{center}
    \includegraphics[scale=0.6]{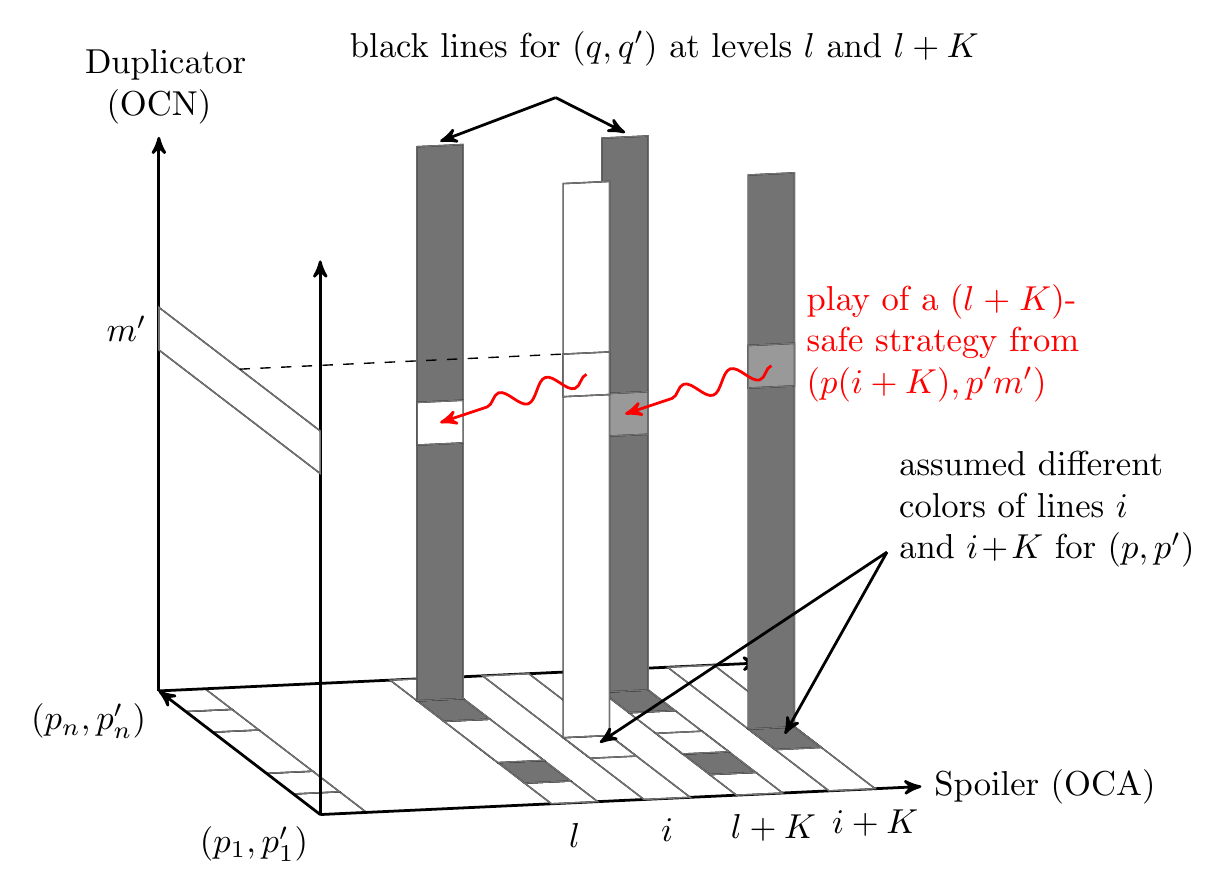}
\end{center}
  \vspace{-1 em}
\caption{Illustrates part 1) of the proof of \cref{lem:dec:regularity}}
\label{fig:dec:regularity}
  \vspace{-1.5 em}
\end{figure}

\Cref{lem:dec:regularity} allows us to fix a value $l\in\N$,
that marks the column in the coloring of $\ssim$ from which on
the distribution of \Scolour\ lines is repetitive
(with period $K$).
We split the simulation relation into two infinite subsets (see
\cref{fig:dec:cut}):
\begin{align*}
    S_{<l} =&\ \ssim\,\cap\ Q\x\{n\in \N\ |\ n<l\}\x Q'\x \N\\
    S_{\ge l} =&\ \ssim\, \cap\ Q\x\{n\in \N\ |\ n\ge l\}\x Q'\x \N.
\end{align*}

\begin{figure}
\begin{center}
    \includegraphics[scale=0.6]{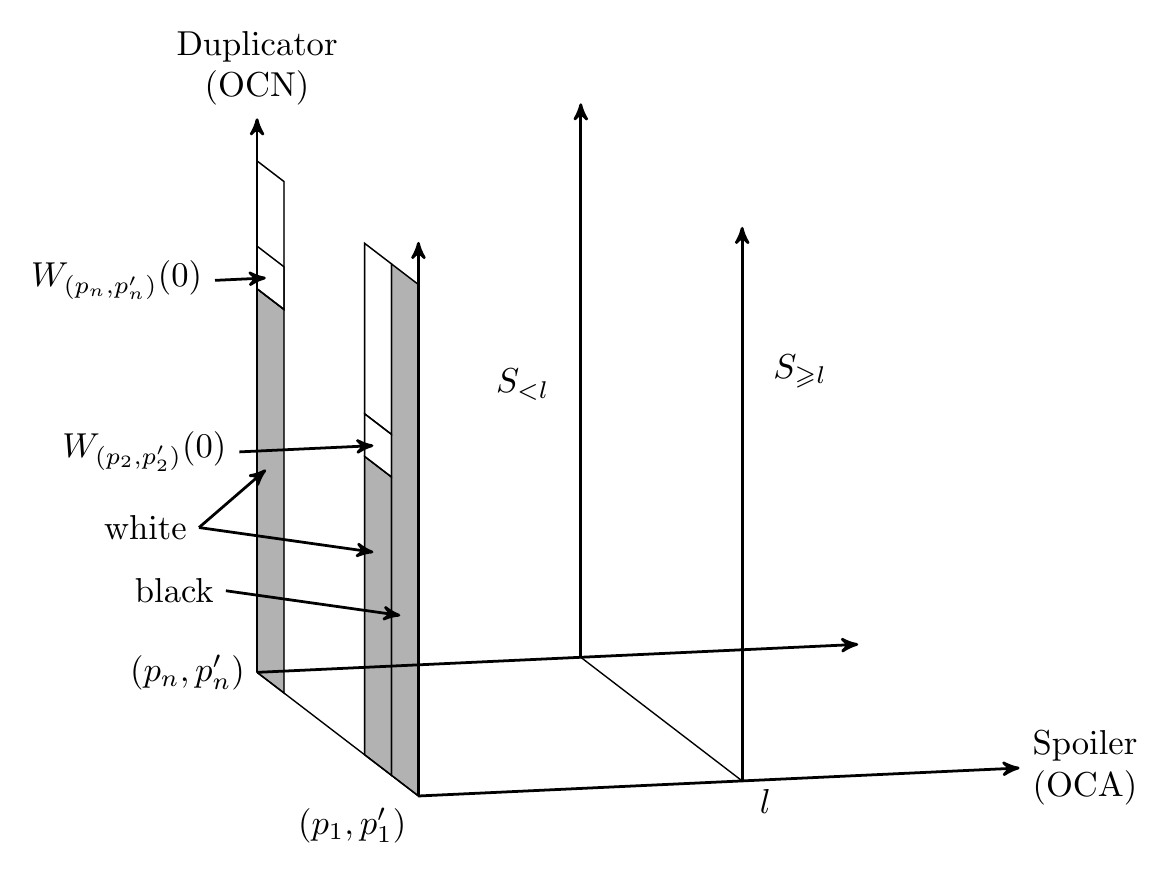}
\end{center}
\vspace{-1 em}
\caption{The coloring of $\ssim$ is cut into two sets
    $S_{<l}$ and $S_{\ge l}$.}
\label{fig:dec:cut}
  \vspace{-1 em}
\end{figure}

Note that the set $S_{<l}$ is semilinear because it is
the upward closure of the minimal \Dcolour\ positions
$\minsuff{p,p'}{i}$ for the finitely many pairs $(p,p')$ and $i<l$.
%
%
%
%
It remains to show that $S_{\ge l}$ is semilinear.
For this, we recall the following result on strong simulation over one-counter nets.

\begin{theorem}[\cite{HLMT:FSTTCS2013}]
    \label{thm:dec:ocn-ocn-slin}
    Let $B$ and $B'$ be two one-counter nets.
    The maximal simulation relation $\SIM{}{B,B'}$
    relative to $B$ and $B'$ is semilinear
    and one can effectively construct a semilinear representation.
%
\end{theorem}

In order to compute $S_{\ge l}$,
we construct two one-counter nets $B$ and $B'$,
such that there is a direct, and Presburger definable, correspondence
between simulation in $B,B'$ and simulation in $A,A'$.
These nets are parameterized by $A,A',l,K$,
the minimal values $\minsuff{p,p'}{l}$ at level $l$
and the patterns $Pat_i$
which determine the distribution of black lines
for indices $l\le i\le l +K$.

%
\begin{lemma}\label{lem:dec:oca-to-ocn}
%
    There exist two one-counter nets $B$ and $B'$ with state-sets
    $R$ and $R'$ respectively,
    and Presburger definable functions
    $F:(Q\x Q'\x\N) \to R$ and
    $G:(Q\x Q'\x\N) \to R'$
    such that 
    for all $p\in Q$, $p'\in Q'$ and $m,m'\in \N$,
    \begin{equation*}
        p(m+l)\SIM{}{A,A'}p'm' \iff
        F(p,p',m) m \SIM{}{B,B'} G(p,p',m) m'.
    \end{equation*}
\end{lemma}
\begin{proof}
    We construct nets $B=(Q\x Q'\x\N_{<K}\cup R_B,\Act',\delta_B)$
    and $B'=((Q\x Q'\x\N_{<K})\cup R_B',\Act',\delta_B')$
    with actions $\Act'=\delta_A\cup \delta_A'\cup\{\$\}$.
    One round of the simulation game w.r.t.~$A$ and $A'$
    will be emulated in two rounds of the game w.r.t.~$B$ and $B'$.
    Apart from auxiliary states in $R_B$ and $R_B'$,
    each state encodes a pair
    of states
    of $A$ and $A'$
    respectively,
    together with the counter value
    of \R\ modulo $K$, so that
    an original
    position
    $(p(m+l), p'm')$ corresponds to the position
    $((p,p',m\mod{K})m,(p,p',m\mod{K})m')$.
    Unless the parameter lets us immediately derive a winner
    for the current position (for example the game reaches a position on a
    \Scolour\ line)
    the new game will continue to
    emulate the old game and end in a position of the above form
    every other round.

    The net $B$ contains the following transitions
    for every $(p,p',m)\in (Q\x Q'\x\N_{<K})$
    and every $t=(p,a,d,q)\in\delta$ where $n=m+d\mod{K}$.
    \begin{align}
        &(p,p',m) \step{t,d} (t,p',n),\\
        &(t,p',n) \step{t',0} (q,q',n)\quad\text{for every } t'=(p',a,d',q')\in \delta
    \end{align}
    The net $B'$ contains
    a universal state $u\in R_B'$ such that $u\step{a,0}u$ for every $a\in\Act'$
    and moreover, the following transitions
    for every $(p,p',m)\in (Q\x Q'\x\N_{<K})$
    and every $t=(p,a,d,q)\in\delta$ where $n=m+d\mod{K}$.
    \begin{align}
        &(p,p',m) \step{t,d'} (t,t',n) &&\text{for every } t'=(p',a,d',q')\in\delta'\\
        &(t,t',n) \step{t',0} (q,q',n)&&\\
        &(q,t',q',n) \step{s,0}u &&\text{for every } s\neq t'\in \delta'.
    \end{align}
    The transitions above allow to emulate one round of the original game
    in two rounds: \R\ announces the transition $t$ she chooses in the original game,
    then \V\ responds by recording his chosen transition $t'$ and remembers both choices in his control-state.
    In the next round \R, who needs to prevent her opponent from becoming
    universal, must faithfully announce her opponents original response.
    Afterwards, \V\ has no choice but to
    also update his state to $(q,q',n)$, which reflects the
    new original pair of states and \R's new counter value modulo $K$.

    So far, this new game is in favor of \V, because it is essentially the
    original game where \R\ is deprived of her zero-testing transitions.
    We now correct this imbalance, using the additional (given) information about the
    values $\minsuff{p,p'}{l}$, as well as the knowledge about which lines
    are completely black.
    
    Recall that for all $i>l$, the line $\Cline{p,p'}{i}$
    is \Scolour\ iff the line $\Cline{p,p'}{i+K}$ is.
    Every state $(p,p',m)$ of $B'$ has a 
    $\$$-labeled self-loop with effect $-1$.
    Moreover, 
    a state $(p,p',m)$ of $B$ has a
    non-decreasing $\$$-labeled self-loop
    if the line $\Cline{p,p'}{m+l}$ is $\Scolour$.
    This ensures that in the new game, \R\ can win from
    positions $(p,p',m\mod{K})m,(p,p',m\mod{K})m'$ if the color of
    $\Cline{p,p'}{l+m}$ is \Scolour, regardless of the actual value $m'$
    of \V's counter.

    Lastly, we add the possibility for \R\ to successfully end
    the game if a position $((p,p',0)l, (p,p',0)m')$
    is reached where her counter value
    equals $l$ and \V's value $m'$ is below $\minsuff{p,p'}{l}$.
    For each $(p,p')\in Q\x Q'$, the control graph of net $B$ contains a
    path
    \begin{equation}
        \label{eq:oca-to-ocn-chain}
    (p,p',0)\step{\$,0}s_k\step{\$,0}s_{k-1}\step{\$,0}\,\cdots\,\step{\$,0}s_0
    \end{equation}
    of length $k=\minsuff{p,p'}{l}$.
    We argue, assuming that nets $B$ and $B'$ are correctly parameterized,
    that
    $p(m+l)\SIM{}{A,A'}p'm'$ iff
    $(p,p',m\mod{K}) m \SIM{}{B,B'} (p,p',m\mod{K}) m'$.
    Observe that this implies the claim of the lemma, as
    projections and multiplication (and division) by fixed values $K$
    are definable in Presburger Arithmetic.

    Assume $p(m+l)\notSIM{}{A,A'}p'm'$ and consider the game on $B$ and $B'$
    from position $((p,p',m\mod{K}) m, (p,p',m\mod{K}) m')$.
    \R\ moves according to her original winning strategy,
    preventing her opponent from reaching state $u$ and thus faithfully
    emulates a play of the game on $A$ and $A'$.
    One of three things must eventually happen:

    1) \V\ is forced to reduce his counter below $0$ and up to then,
    \R's counter always remains strictly above level $l$ and no visited position
    corresponds to a point on a black line.
    Such a play is losing for \V\ in both games.

    2) A position 
    $((q,q',n\mod{K})n,(q,q',n\mod{K})0)$ is reached
    and the line $\Cline{q,q'}{l+n}$ is \Scolour.
    This means the states $(q,q',n\mod{K})$
    of both nets have a $\$$-labeled self loop,
    where only the one in $B'$ is decreasing. \R\ wins from such a position
    by iterating this loop.

    3) A position 
    $((q,q',0)0,(q,q',0)m')$ is reached where \R's counter equals $0$
    and $m'<\minsuff{q,q'}{l}$. In this case, \R\ wins by
    moving along the $\$$-labeled path described by
    \cref{eq:oca-to-ocn-chain},
    which allows her to make exactly $\minsuff{q,q'}{l}+1$
    many steps which are decreasing \V's counter.

    Conversely, assume $(p,p',m\mod{K}) m \notSIM{}{B,B'} (p,p',m\mod{K}) m'$
    and consider a play along which \R\ plays according to some winning strategy.
    Until \R\ makes use of a $\$$-labeled step, the
    play in the game on $B,B'$ directly corresponds to a play
    on $A,A'$. This is because \R\ must prevent her opponent from
    reaching the universal state $u$.
    A play that \R\ wins without $\$$-steps
    thus yields a winning play in the game on $A,A'$.
    Otherwise, consider a first position
    $((q,q',n\mod{K})n,(q,q',n\mod{K})n')$ from which \R\ plays
    a $\$$-step.
    The fact that the state $(q,q',n\mod{K})$ has a $\$$-labeled outgoing
    transition means that either
    $\Cline{q,q'}{n+l}$ is \Scolour\ or this line is \Dcolour\ but
    $n\mod{K}=0$.
    The former means that in particular that $q(n+l)\notSIM{}{A,A'}q'n'$.
    In the latter case, \R\ moves from state $(q,q',0)$ to
    the initial state $s_k$ of a chain of length $k=\minsuff{q,q'}{l}$,
    at the end of which she would deadlock.
    This does not happen because the play is a win for \R.
    Therefore, the counter value $n'$ of \V\ must be
    strictly below $\minsuff{q,q'}{l}$.
    Since $K$ divides $n$,
    point~2 in \cref{lem:dec:regularity} implies
    $n'<\minsuff{q,q'}{l}\le \minsuff{q,q'}{l + n}$
    which means that
    $q(l + n) \notSIM{}{A,A'}q'n'$.
    
    We conclude that a winning strategy
    from $((p,p',m\mod{K})m, (p,p',m\mod{K})m')$ in the game on $B,B'$
    prescribes a strategy for \R\ in the game on $A,A'$ from position
    $(p(m+l),p'm')$ such that each play eventually
    leads to a winning position
    and therefore, $p(m+l)\notSIM{}{A,A'}p'm'$
\end{proof}

The semilinearity of $S_{\ge l}$
and thus of $\ssim$ now follows
from \cref{thm:dec:ocn-ocn-slin} and \cref{lem:dec:oca-to-ocn}.
This completes the proof of \cref{thm:dec:sim-is-slin}.

\section{Undecidability Results}

\newcommand{\qinit}{\ensuremath{q_{init}}}
\newcommand{\qhalt}{\ensuremath{q_{halt}}}
\newcommand{\inci}{\ensuremath{c^+_i}}
\newcommand{\deci}{\ensuremath{c^-_i}}
\newcommand{\z}{\ensuremath{\mathbf{z}}}
\newcommand{\nz}{\ensuremath{\mathbf{nz}}}
\newcommand{\incm}[3]{(\ensuremath{{#1},[{#2}^+ ; {#3}])}}
\newcommand{\decm}[4]{\ensuremath{({#1},[\mathbf{if}~({#2}=0)~\mathbf{then}~{#2};{#3}~\mathbf{else}~{#2}^-;{#4}]) }}

\subsection{Pushdown Energy Games}

We show that the fixed as well as the unknown initial credit problem for 
pushdown energy games is undecidable. This holds even
with a single energy dimension. 
Consequently, we may use \cref{lem:dec:reverse-reduction} to deduce 
the undecidability of the simulation problem between pushdown automata and OCN.

The undecidability is established 
by a reduction from the  halting  problem for Minsky's $2$-counter machines.
We begin by recalling the definition.
The machine is equipped with two counters, $c_1$ and $c_2$,
that take values over natural numbers. There are two types
of transitions. The first kind simply increments one of the counters. 
The second kind checks the value of a counter and 
decrements it iff it is $>0$.

\begin{definition}
A  $2$-counter machine (MCM) $M$ is a tuple $(Q,\qinit,\qhalt,\delta)$ where $Q$ is a finite set of control-states, $\qinit$ is the initial state, $\qhalt$ 
is the halting state, and $\delta$ is a finite set of transitions of
the following two forms.
\begin{itemize}
\item $\incm{q}{c_i}{q'}$, $i \in \{1,2\}$ and $q \neq \qhalt$. Increment counter $c_i$ unconditionally and go to state $q'$.
\item $\decm{q}{c_i}{q'}{q''}$, $i \in \{1,2\}$ and $q \neq \qhalt$. Check the value of counter $c_i$, go to state $q'$ if it
is $0$ or go to state $q''$ after decrementing $c_i$ if it is $>0$.
\end{itemize}

A \emph{configuration} of such a machine is an element of $Q\times\N^2$.
The initial configuration is $(\qinit,(0,0))$. We say that the configuration
$(q,(m_1,m_2))$ moves to $(q',(m_1',m_2'))$ in one step, written 
$(q,(m_1,m_2))\step{} (q',(m_1',m_2'))$, iff
        \begin{enumerate}
          \item $\incm{q}{c_i}{q'} \in \delta$, $m'_i=m_i+1$  and $m'_{3-i}=m_{3-i}$
          \item $\decm{q}{c_i}{q'}{q''} \in \delta$, $m_i=m'_i=0$ and $m'_{3-i}=m_{3-i}$.
          \item $\decm{q}{c_i}{q''}{q'} \in \delta$, $m_i>0$, $m'_i=m_i-1$ and $m'_{3-i}=m_{3-i}$.
        \end{enumerate}

A \emph{run} is a finite or infinite sequence of steps between configurations.
A run is \emph{maximal} if it is either infinite or ends at a configuration
where no move is possible. W.l.o.g. we assume that at least one move
is possible in any configuration whose control state is not $\qhalt$.

A $2$-counter machine is said to be \emph{deterministic} if there is a unique 
maximal run starting from the initial configuration $(\qinit,(0,0))$. 
Notice that this  run either reaches the state $\qhalt$ (and halts) or 
is  infinite.
\end{definition}

We now prove the following lemma which immediately implies that both
the fixed and unknown initial credit problem for pushdown games are undecidable.

\begin{lemma}
Given a deterministic MCM $M=(Q,q_{init},q_{halt},\delta)$, one can 
effectively construct a $1$-dimensional pushdown energy game 
$G = (Q_0, Q_1, \Gamma,\delta_G,1)$ s.t.~$M$ halts iff Player $0$ wins the energy game from every initial
energy credit. Moreover, if Player $0$ wins the energy game for some initial energy credit 
then she wins from every initial energy credit.
\end{lemma}
\begin{proof}
If $M$ does not terminate then it diverges i.e., the only valid run is infinite.
The overall idea is to let Player 1 propose this infinite run by pushing the corresponding sequence $\tau_1,\tau_2,\ldots$ of transitions onto the stack.
If $M$ actually terminates and there is no infinite run, Player 1
must eventually ``cheat'' and announce a next step that is not a valid
continuation of the run committed to the stack.
After each such move, Player 0 can choose to either accept the last MCM step
and let Player 1 continue to push the next one, or she can choose
to challenge its validity and move the game to a test.
There is a test for every type of
error that can be spotted by Player $0$
and each such test has one of three possible outcomes.

\begin{enumerate}
\item The energy level goes below zero and Player $0$ wins.
\item The game returns to a position with initial state and empty stack,
    but the energy level is smaller than it was before.
\item The game returns to a position with initial state and empty stack,
    but the energy level is greater or equal than it was before.
\end{enumerate}

The first outcome is obviously good for Player $0$
and so is outcome $2$, because
eventually, after sufficiently many such outcomes, the energy has to run out
and she will win.
Outcome $3$ is good for Player $1$, because it makes his position at least as good
as it was before.
We will implement test gadgets that Player $0$ may choose to invoke.
If Player $0$ has correctly spotted an error then the outcome of the gadget 
will be $1$ or $2$, and if she was wrong then the outcome will be $3$. 

We now describe the construction formally, starting with the part of the game in which
Player $1$ writes a run of the MCM onto the stack.
%
The first part of the construction guarantees that the ending state of each transition 
$\tau_i$ matches the starting state of the next transition $\tau_{i+1}$.
Given $M$, we build a finite graph as follows.
Vertices are states of $M$ and for every transition of the first type $\tau_l=\incm{q}{c_i}{q'}$, we add a directed edge $q, q'$ labeled with $ \tau_{l,c_i} $ and for every transition of the second type
 $\tau_l=\decm{q}{c_i}{q'}{q''}$ we add a pair of directed edges $q,q'$ and $q,q''$ labeled $\tau_{l>0, c_i}$ and $\tau_{l=0, c_i}$, respectively. 
Every path in this graph corresponds to some correct or incorrect run of $M$.
Incorrectness may come from the fact that paths in the graph do not care about the 
values of the counters. 

Now we encode this graph into a part of the energy game.
Vertices becomes states owned by Player $1$, and each edge $q,q'$ labeled with 
$\tau_{X}$ is encoded by two sets of transitions. In the first set we 
have transitions $(q, Y ,s_{\tau_{X}},\tau_{X} Y, 1)$ for every stack symbol $Y$,
and $s_{\tau_{X}}$ is an intermediate state which belongs to Player $0$.
These push a record of the transition $\tau_{X}$ onto the stack and increment the
energy by $1$.
In the second set we have transitions $(s_{\tau_{X}}, Y , q', Y, 0)$ 
for every stack symbol $Y$. These do not change the stack and energy.
Additionally, in $s_{\tau_{X}}$, Player $0$ can decide if she wants to follow the edge to $q'$ or if she would 
rather invoke some testing gadget (see below).

In the halting state we add a self-loop which decrements the energy; 
this guarantees that if $M$ halts then Player $0$ wins.  

The crucial properties are that (except in the situation when the game reaches the halting state):  
\begin{itemize}
\item on the stack we have a word which remembers the path from the beginning up to now,
\item the energy level is equal to the initial value of the energy  $+$ the number of elements on the stack.
\end{itemize}

So what are the possible errors? Player $1$ can cheat only when he has a choice. 
Since $M$ is deterministic, we know that in the game choice comes only from the translation of 
the transitions of the second type $(\decm{q}{c_i}{q'}{q''})$ .
There are two types of errors.
\begin{enumerate}
\item Player $1$ tries to go from $q$ to $q'$ when $c_i >0$, or
\item Player $1$ tries to go from $q$ to $q''$ when $c_i =0$.
\end{enumerate} 

To detect these errors, we define four gadgets; two types times two counters. 
The gadget that we have to use is determined by the last transition; if it pushes to the stack $\tau_{l=0, c_i}$ then we use the first type gadget for the counter $c_i$ and if it pushes $\tau_{l>0, c_i}$ then we use the second 
type gadget for $c_i$.        

In both gadgets we pop stored transitions from the stack, and reduce the energy level accordingly,
in order to check the conditions $c_i >0$ and $c_i =0$, respectively.
\begin{itemize}
\item 
In the first type of error, Player $0$ should gain (outcome 1 or 2) 
if $c_i>0$, i.e., if in the stored history the number of increments 
of counter $c_i$ is greater than the number of decrements of $c_i$.
Otherwise, we should get outcome 3.
\item 
In the second type of error, Player $0$ should gain if $c_i=0$, i.e., 
we need to check if the number of decrements is equal to the number of increments.
Observe that we can safely assume that it is not greater, because this would mean
$c_i < 0$ and then there was an error earlier.
So Player $0$ should gain (outcome 1 or 2) if the number of decrements is equal to the number of increments,
but lose (outcome 3) if the number of decrements is less than the number of increments.
\end{itemize}
In the first type of gadget, we allow Player $0$ to pop the content of the stack 
according to following rules. 
\begin{enumerate}
\item In removing a transition not affecting $c_{i}$ she reduces
the energy level by $1$.
\item In removing a transition incrementing $c_i$ she reduces
the energy level by $2$.
\item In removing a transition decrementing $c_i$ she leaves
the energy level unchanged.
\end{enumerate}
If the number of increments was greater than the number of decrements 
then in the end of the gadget the energy level drops below
the initial energy level (outcome 1 or 2). 
This is because, before we enter the gadget, 
the energy level is equal to the initial energy level $+$ the number of elements on the stack, 
and in the gadget we decrease it by more than the number of elements on the stack. 
On the other hand, if the numbers of decrements and increments are equal, 
then the energy level is equal to the initial one (outcome 3).

In the second type of gadgets, the roles are reversed.
\begin{enumerate}
\item In removing a transition not affecting $c_i$ she reduces
the energy level by $1$.
\item In removing a transition incrementing $c_i$ she leaves
the energy level unchanged.
\item In removing a transition decrementing $c_i$ she reduces
the energy level by $2$.
\end{enumerate}
Finally, in the end we decrement the energy by $1$.

If the number of decrements is equal to the number of increments then we end with 
an energy level which is smaller than the initial energy level (outcome 1 or 2), due to
the final decrement by $1$.
On the other hand, if the number of 
decrements is smaller than the number of increments, 
then the final energy level is greater than (or equal to) 
what it was in the beginning (outcome 3).
\end{proof}
\subsection{One-counter energy games}\label{subsec:undec-onecounter}

We show that one-counter energy games of energy dimension $n \ge 2$ are
undecidable, via the undecidability of the corresponding simulation games. 

\begin{theorem}\label{thm:undec-sim-ocn-vass}
Simulation preorder between OCN and VASS of dimension $\ge 2$ is 
undecidable in both directions.
\end{theorem}
\begin{proof}
Consider a deterministic Minsky 2-counter machine $M$ 
with a set of control-states $Q$, counters $c_1$ and $c_2$ and
initial configuration $(q,(0,0))$. 
It either eventually reaches the accepting state $q_{\it halt}$ or runs forever.
We construct an OCN $A=(Q_A,\Act,\delta_A)$ with initial configuration
$(q,0)$ and a VASS $V=(Q_V,\Act,\delta_V)$ of dimension $2$ with
initial configuration $(q,(0,0))$ such that
$M$ halts iff $(q,0) \not\ssim (q,(0,0))$.
(The construction for the other simulation direction is very similar.)

Let $\Act = \{a,z,{\it nz},c,h\}$, $Q_A = Q \cup Q_A'$ and $Q_V = Q \cup Q_V'$
where $Q_A', Q_V'$ contain some auxiliary control-states (see below).
In the simulation game we maintain the following invariant of game configurations
$((q,z),(q',(x,y)))$. If $q=q' \in Q$ then $z=x+y$. I.e., except in some auxiliary states
of $Q_A', Q_V'$, the OCN counter will contain the sum of the VASS counters.
The idea is that the simulation game emulates the computation of $M$, where the
two counters are stored in the VASS counter values $x$ and $y$, respectively.
Via the classic forcing technique, \V\ gets to choose the next transition of $M$.
The only possible deviation from a faithful emulation of $M$ is where \V\ chooses
a zero-transition of $M$ when the respective counter contains a nonzero value,
e.g., $x>0$. In this case \R\ can win the game by forcing a comparison of
$z$ with $y$. If $x>0$ then, by the above invariant, $z>y$ and \R\ wins.
Otherwise, if $x=0$ then $z=y$ and \V\ wins.

Since $M$ is deterministic, there is only one transition rule for every 
control-state $q$. 

If the rule of $M$ is of the form $\incm{q}{c_1}{q'}$
then we add a rule $(q,a,+1,q')$ to $\delta_A$
and a rule $(q,a,q',(1,0))$ to $\delta_V$. 
(The other case where $c_2$ is incremented is symmetric.)
Thus one round of the simulation game
emulates the transition of $M$ and the invariant is maintained.

Otherwise, the rule of $M$ is of the form 
$\decm{q}{c_1}{q'}{q''}$
(the case where $c_2$ is tested is symmetric).
We add a rule $(q,a,0,\hat{q})$ to $\delta_A$
and rules $(q,a,q_1,(0,0))$ and $(q,a,q_2,(-1,0))$ to $\delta_V$.
By choosing $q_1$ (resp.\ $q_2$) \V\ claims that $c_1$ is zero (resp.\ nonzero).
A nonzero claim is certainly correct, since the transition to $q_2$ decrements
the counter. However, a zero claim might be false.
Now \R\ can either accept this claim or challenge it (and win iff it is false).
For the case where a nonzero claim is accepted we
add a rule $(\hat{q},{\it nz},-1,q'')$ to $\delta_A$ and a rule
$(q_2,{\it nz},q'',(0,0))$ to $\delta_V$. 
Similarly for a case where a zero claim is accepted we
add a rule $(\hat{q},z,0,q')$ to $\delta_A$ and a rule
$(q_1,z,q',(0,0))$ to $\delta_V$.
In either case, two rounds of the simulation game
emulate the transition of $M$ and the invariant is maintained.
Since \R\ must not spuriously accept a choice that \V\ has not made,
we add transitions $(q_2,z,U,(0,0))$, $(q_1,{\it nz},U,(0,0))$ 
and $(U,\alpha,U,(0,0))$ to $\delta_V$ for every $\alpha \in Act$.
Here \V\ goes to the universal state $U$ and wins the simulation game.

As explained above, a nonzero claim by \V\ is always correct and thus
cannot be challenged. The following construction implements a challenge by 
\R\ to a zero claim of \V. We add a transition
$(\hat{q},c,0,q_c)$ to $\delta_A$ and transitions
$(q_1,c,q_c,(0,0))$ and $(q_2,c,U,(0,0))$ to $\delta_V$. 
If \R\ spuriously issues a challenge to a zero claim that \V\ has not
made (where he is in state $q_2$) then \V\ goes to the universal state $U$ 
and wins. Otherwise, both players are in state $q_c$ and the simulation game 
is in state $((q_c,z),(q_c,(x,y)))$ for $z=x+y$ by the invariant above.
The challenge is evaluated by the following rules.
We add a rule $(q_c,c,-1,q_c)$ to $\delta_A$ and a rule
$(q_c,c,q_c,(0,-1))$ to $\delta_V$. 
If the zero claim by \V\ was false then $x>0$ and thus $y < z$.
Therefore, \R\ wins the simulation game from $((q_c,z),(q_c,(x,y)))$,
because eventually the second counter of \V\ reaches zero before \R.
However, if the zero claim by \V\ was true then $x=0$ and $y=z$ and
both players reach zero (and get stuck) at the same time, and thus \V\ wins.

Finally, we add a transition $(q_{\it halt}, h, 0, q_{\it halt})$ to $\delta_A$,
i.e., the state $q_{\it halt}$ is winning for \R\ by a special action $h$. 

To summarize, if $M$ does not halt then \V\ wins the simulation game by a 
faithful emulation of the infinite computation of $M$, 
because every challenge by \R\ will also lead to a 
win by \V. 
Conversely, if $M$ halts then every faithful emulation would lead to state 
$q_{\it halt}$ and a win by \R. The only possible deviation from a faithful
emulation is a false zero claim by \V. However, the challenge 
construction above ensures that \R\ can also win in this case.
Thus \R\ wins the simulation game iff $M$ halts.
\end{proof}

Since OCA subsume OCN (by Def.~\ref{def:oca-ocn}), 
this implies that simulation preorder between
OCA and VASS of dimension $\ge 2$ is also undecidable.
Thus, with Lemma~\ref{lem:dec:reverse-reduction}, we obtain the following
theorem.

\begin{theorem}\label{thm:undec-onecounter}
The fixed initial credit problem is undecidable for one-counter energy games
of energy dimension $n \ge 2$.\\
I.e., given a one-counter energy game $G=(Q_0, Q_1,\delta,\delta_0,n)$ of dimension $n \ge 2$,
it is undecidable whether a configuration $(q,k,E)$ is winning for Player $0$.
\end{theorem}


\section{Conclusion and future work}\label{sec:conclusion}

Our decidability results for infinite-state energy games show a 
surprising distinction between pushdown automata and one-counter automata.
While pushdown energy games are undecidable even for the simplest case of a
1-dimensional energy, the decidability border for one-counter energy games
runs between the cases of 1-dimensional and multi-dimensional energy.

Some questions for future work concern the decidability of the {\em unknown
initial credit problem} for infinite-state energy games.
We have shown the undecidability of this problem for pushdown energy games,
but it remains open for one-counter energy games.
While we have shown that the winning sets of 1-dimensional one-counter energy games are
semilinear, our proof does not yield an effective procedure for
constructing these semilinear sets (which would immediately imply the decidability
of the unknown initial credit problem).

In multidimensional one-counter energy games, the winning sets 
are certainly not semilinear (even though they are 
upward-closed w.r.t.\ the energy). Otherwise, one could enumerate semilinear sets
and effectively check (by Presburger arithmetic) 
whether they are winning sets containing the initial configuration, 
and thus obtain a positive semi-decision procedure. 
Together with the obvious negative semidecidability (by expanding the game tree)
this would yield an impossible decision procedure for the fixed initial credit
problem.
In spite of this, the unknown initial credit problem could still be decidable.

The unknown initial credit problem for energy games is closely related to
limit-average games.
While even 1-dimensional limit-average games are undecidable for pushdown
automata \cite{Chatterjee-Velner:LICS2012}, 
the decidability of (multi-dimensional) limit-average games on one-counter
automata is open.

\section{Acknowledgement}
Piotr Hofman acknowledges a partial support by the Polish NCN grant
2012/05/N/ST6/03226.\\
Richard Mayr, Parosh Aziz Abdulla, Mohamed Faouzi Atig and Patrick Totzke 
acknowledge partial support by UK Royal Society grant IE110996.


\end{document}